\documentclass{article}[11pt]
\usepackage{graphicx} 
\usepackage{amssymb,enumerate,graphicx,verbatim,natbib,amsthm,bbm}
\usepackage{todonotes}
\usepackage[inline, shortlabels]{enumitem}
\usepackage{thmtools}
\usepackage{thm-restate}
\usepackage[hidelinks]{hyperref}
\usepackage{subcaption}
\usepackage{mwe}
\usepackage[toc,page,header]{appendix}
\usepackage{minitoc}
\usepackage{algorithmic}
\usepackage[linesnumbered,ruled,vlined]{algorithm2e}
\usepackage[left=3.0cm, right=3.0cm, top=3cm]{geometry}
\usepackage{amsmath}
\usepackage[nameinlink,noabbrev,capitalise]{cleveref}
\Crefname{table}{Table}{Tables}
\Crefname{assumption}{Assumption}{Assumptions}
\crefname{equation}{}{}
\newtheorem{theorem}{Theorem}[section]

\usepackage{varwidth}
\usepackage{authblk}

\usepackage{xcolor}

\def\commenton{1}
\newcommand{\HNcomment}[1]{\if\commenton1{\color{red}(HN: #1)}\fi}

\title{\vspace{-2cm} Breaking the Winner's Curse with Bayesian Hybrid Shrinkage}

\author{Richard Mudd}
\author{Rina Friedberg}
\author{Ilya Gorbachev}
\author{Houssam Nassif}
\author{Abbas Zaidi}

\affil{Meta Inc \authorcr
  \{\tt \normalsize rmudd, rinafriedberg ilyagorbachev, houssamn, abbaszaidi\}@meta.com}
\date{}

\begin{document}

\maketitle
\section{Introduction}

A ``Winner's Curse'' arises in large-scale online experimentation platforms~\citep{lee2018winner} when the same experiments are used to both select treatments and evaluate their effects. In these settings, classical difference-in-means estimators of treatment effects are upwardly biased and conventional confidence intervals are rendered invalid~\citep{andrews2024inference}. The bias scales with the magnitude of sampling variability and the selection threshold. It decreases inversely with the treatment's true effect size~\citep{van2021significance}. 

Bayesian methods that leverage empirical priors have been proposed and demonstrated, in specific settings, to yield superior inferential properties under selection~\citep{ejdemyr2024estimating, kessler2024overcoming} when compared to the standard ``Face Value'' estimators. However, Bayesian analysis is sensitive - and arguably especially susceptible under selection~\citep{rasines2022bayesian} - to the choice of prior, and models that require computationally expensive numerical integration techniques are ill-suited for at-scale deployment~\cite{Li2027optimalPolicies}.

We propose a new Bayesian approach that incorporates experiment-specific ``local shrinkage'' factors that mitigate sensitivity to the choice of prior and improve robustness to assumption violations. Crucially, we demonstrate how the associated posterior distribution can be estimated without numerical integration techniques, making it a practical choice for at-scale deployment. 

\section {Modeling Selection in Bayesian Analyses}  


In a frequentist analysis, the Winner's Curse is an artifact of the sampling distribution of a statistic being altered by selection. This implies that an explicit correction for selection via a selection model is always required~\cite{andrews2024inference}. 

The need for a selection model under the Bayesian paradigm, is less obvious. One longstanding perspective is that selection is not required because posterior distributions are already conditioned on the data~\cite{dawid1994selection}: If a different sample of data was observed under a hypothetical replication, this is irrelevant with respect to the posterior at hand. This conditioning does not necessarily account for the selection mechanism that led to those observations though.

The interaction of the selection mechanism with the parameter space determines whether such an adjustment is needed in a Bayesian analysis~\citep{yekutieli2012adjusted}. If we consider parameters and data to be sampled from a joint distribution, no explicit selection adjustment is required. On the other hand, if we consider a parameter to be sampled from its marginal distribution and held fixed, with the data then sampled repeatedly from the associated conditional distribution, an adjustment is required. This abstraction is best illustrated by example~\cite{woody2022optimal}:

\begin{itemize}
\item{Suppose we consider a set of experiments that test different changes to a product. In each experiment, the data collected is used to make a launch decision according to some predetermined criteria. For experiments that launched, we want to estimate the effect size. In this context a Bayesian approach would \textit{not} require adjustment for selection.}
\item{Suppose we consider a single proposed change to a recommender system, that we decide to test multiple times in different experiments. If we are interested in using the set of experiments that pass some launch criteria to estimate the size of the proposed change, we \textit{would} need to adjust for selection.} 
\end{itemize}

Different settings leading to different solutions arise from the likelihood principle: likelihoods that differ only by a scalar multiplier are equivalent, and should lead to same inferences. 

\begin{theorem} \label{theorem:jointselection}
With joint selection, the posterior distribution under selective inference is equivalent to its unadjusted counterpart, eliminating the need for a separate selection model.
\end{theorem}

\begin{proof}
Under joint selection, let the joint distribution of $\{\hat{\theta}, \theta\}$ be given by $\pi_{S}(\hat{\theta}, \theta) = \pi(\theta)\cdot f(\hat{\theta}|\theta)\cdot \frac{1(\hat{\theta}\in S)}{\pi(S)}$, where the marginal selection probability is $\pi(S) = \int\int_{S}\pi(\hat{\theta}|\theta)\pi(\theta)\mathrm{d}\hat{\theta}\mathrm{d}\theta$. The marginal density of the data under selection is $m(\hat{\theta}) = \int f(\hat{\theta}|\theta)\pi(\theta)\mathrm{d}\theta$. Truncated to the selected sample $S$, this marginal density is $\int 1(\hat{\theta}\in S)\pi(\theta)\frac{f(\hat{\theta}|\theta)}{m(S)}\mathrm{d}\theta = \frac{m(\hat{\theta})}{m(S)}\cdot 1(\hat{\theta}\in S)$. Next consider the posterior under selection ${\pi_{S}(\hat{\theta}|\theta}) = \frac{\pi_{S}(\hat{\theta}, \theta)}{m_{S}} = \frac{\pi(\theta)\cdot f(\hat{\theta}|\theta)}{m(\hat{\theta})} = \pi(\theta|\hat{\theta})$ which is identical to the unadjusted posterior.
\end{proof}

Many applied settings can reasonably be described by the joint distribution paradigm, meaning that a well-motivated prior is often sufficient to overcome the Winner's curse. In this context, strategies to validate the modeling choice and ensure a well-calibrated prior distribution are vital.
\vspace{0.3cm}
\section{Formulation}  

\subsection{A Bayesian Model for Inference Under Selection}  

Consider a collection of $N$ experiments indexed by $i = 1, \ldots, N$. Each experiment consists of $m$ units, indexed by $j = 1, \ldots, m$, which are assigned to treatment conditions denoted by $Z_{ji}$. $Z_{ji} = 1$ indicates that unit $j$ in experiment $i$ receives the treatment, while $Z_{ji} = 0$ indicates control. The outcome of interest for each unit is $Y_{ji}$. Our estimand for each experiment $i$ is the ratio of the expected potential outcomes under treatment and control, expressed as $\theta_{i} = \frac{E[Y(1)]}{E[Y(0)]}$. 

The classical or \textit{``Face Value''} estimator, $\hat{\theta}_{i}^{FV}$, is computed as the ratio of sample averages of observed outcomes in treated and control groups:
\[
\hat{\theta}_{i}^{FV} = \frac{\displaystyle\frac{1}{\sum_{j=1}^{m} Z_{ji}} \sum_{j=1}^{m} Y_{ji} Z_{ji}}{\displaystyle\frac{1}{\sum_{j=1}^{m} (1 - Z_{ji})} \sum_{j=1}^{m} Y_{ji} (1 - Z_{ji})}
\]
with standard error $\hat{\sigma}_{i}$. This estimator is known to be biased in the presence of selection effects.

To address this, we propose a new \textit{``Bayesian Hybrid Shrinkage''} approach, which can be formulated as a hierarchical post-hoc model:
\begin{align*}
    \hat{\theta_{i}}|\theta_{i}, \hat{\sigma_{i}}^{2} &\sim \mathrm{N}(\theta_{i}, \hat{\sigma^{2}_{i}}),\\
    \theta_{i} | m_{0}, \lambda_{i}, \tau &\sim \mathrm{N}(m_{0}, \lambda_{i}\cdot\tau),\\
    \lambda_{i}|a, b &\sim \mathrm{InverseGamma}\left(\frac{a}{2}, \frac{b}{2}\right).
\end{align*}

True effect $\theta_{i}$ is assigned a normal prior with mean $m_{0}$ and variance composed of a global scale parameter $\tau$ modulated by local shrinkage factor $\lambda_i$, which allows the model to adaptively shrink estimates differently across experiments. $\lambda_i$ is an inverse-Gamma hyperprior parameterized by $a$ and $b$, which control the shrinkage distribution across experiments. The hierarchical structure balances borrowing strength across experiments with flexibility to accommodate experiment-specific variability.

\vspace{0.3cm}

\subsection{Inference}  
Inference can be carried out either via posterior simulation—using iterative sampling over the target and nuisance parameters—or analytically by fixing $\lambda_i$ at its posterior mode.
The posterior distribution of the true effect $\theta_i$, conditional on the observed estimator $\hat{\theta}_i$, the local shrinkage factor $\lambda_i$, and the global scale parameter $\tau$, is given by:
\begin{align}
\theta_i \mid \hat{\theta}_i, \lambda_i, \tau \sim \mathrm{N}\left(
\frac{\hat{\sigma}_i^2}{\hat{\sigma}_i^2 + \lambda_i \tau} m_0 + \frac{\lambda_i \tau}{\hat{\sigma}_i^2 + \lambda_i \tau} \hat{\theta}_i, \quad
\left(\frac{1}{\hat{\sigma}_i^2} + \frac{1}{\lambda_i \tau}\right)^{-1}
\right).
\end{align}
A special case of this posterior arises when the local shrinkage factor is fixed at $\lambda_i = 1$ for all experiments, corresponding to a Bayesian estimator that imposes only global shrinkage~\citep{kessler2024overcoming}. We refer to this as the \textit{``Bayesian Global Shrinkage''} model.

\subsection{Validation via Predictive Checking}\label{section:validation}
Given the complexities of validation in experimental settings, we take inspiration from \textit{predictive checking} \cite{rubin1998more, gelman1996posterior}. This type of assessment uses predictive simulation of new data under a model of interest $M$ that we denote as $\hat{\theta}_{rep}$, along with a related statistic $T(\cdot)$ to characterize discrepancies against its observed counterpart $g(T(\hat{\theta}_{rep}), T(\hat{\theta}))$. This discrepancy can be used to assess any aspect of the model we want to validate (e.g., prior parameter choices or decision boundary) facilitated by the construction of a reference distribution,
$p(g(T(\hat{\theta}_{rep}), T(\hat{\theta}))|M, \hat{\theta})$.

These flexible constructions enable assessment of quantities like coverage for uncertainty intervals or goodness-of-fit of the posterior distribution. Specifically, to understand goodness-of-fit, one intuitive quantity is the tail area probability $g(T(\hat{\theta}_{rep}), T(\hat{\theta})) = T(\hat{\theta}_{rep}) \geq T(\hat{\theta})$, an analogue to a p-value that can be computed for a given model averaged over the posterior distribution of the parameter of interest $p(\theta_{i}|M, \hat{\theta})$. This concept can be augmented by data splitting strategies such as replication studies or data fission~\cite{leiner2023data}. We utilize this concept in the following section to assess performance first in simulation and then empirically from real experiment data.
\vspace{0.3cm}
\section{Results}


\subsection {Simulated Experiments}  
To evaluate the performance of the proposed \textit{Bayesian Hybrid Shrinkage} approach, we conduct simulated experiments comparing it against the \textit{Face Value} and \textit{Bayesian Global Shrinkage} approaches. These simulations are designed to reflect plausible real-world scenarios where the analysis prior is misspecified in various ways:
\begin{enumerate}
    \item \textbf{Misspecified mean:} The analysis prior mean $m_0 = 0$ differs from the true effect size distribution mean, where $\theta_i \sim \mathrm{N}(\mu, \epsilon)$.
    \item \textbf{Heavy-tailed distributions:} The true effect sizes follow a t-distribution with $\nu$ degrees of freedom, $\theta_i \sim t_\nu(\mu, \epsilon)$, which has heavier tails than the assumed normal analysis prior. 
    \item \textbf{Hidden Selection:} True effects are drawn as two-dimensional vectors from a bivariate normal distribution, $\boldsymbol{\theta_i} \sim \mathrm{N}_2(\boldsymbol{\mu}, \boldsymbol{\Sigma})$, where the covariance matrix $\boldsymbol{\Sigma}$ has correlation $\rho$. Selection is applied jointly on both dimensions, but the analysis considers only one target parameter, testing robustness to unmodeled correlated selection.
\end{enumerate}
The results of these simulations are summarized in Figure~\ref{fig:misspecified}, which compares the three approaches across the three simulation settings and three performance metrics: \textit{Mean Squared Error (MSE)}, \textit{Bias}, and \textit{Coverage Probability} of 90\% uncertainty intervals. As expected, the \textit{Bayesian Global Shrinkage} approach performs optimally when the prior is correctly specified. However, the \textit{Bayesian Hybrid Shrinkage} approach is shown to be more robust to misspecification and consistently outperforms the \textit{Face Value} approach in all settings considered.

\begin{figure*}
    \centering
    \subfloat{
        \includegraphics[width=0.3\textwidth]{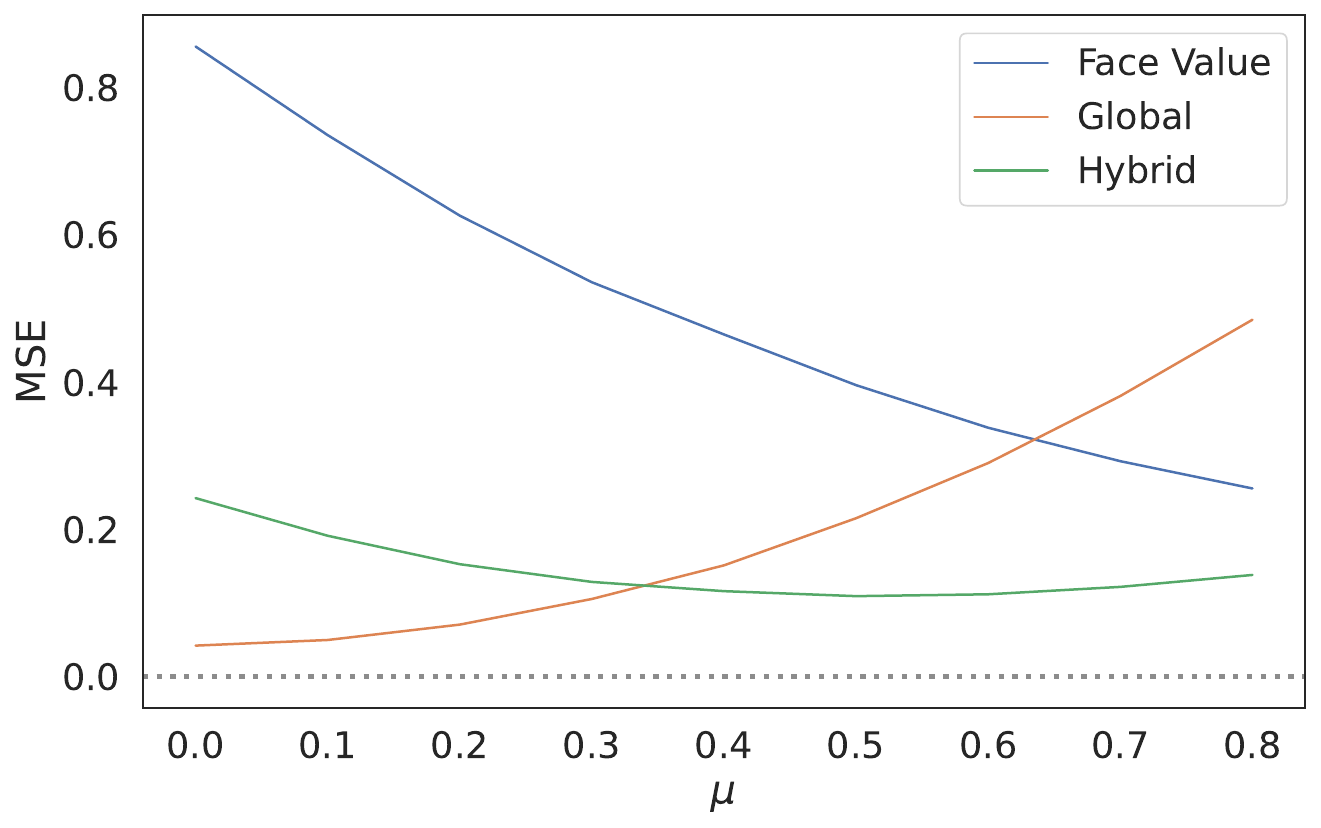}
        \label{subfig:a}
    }\hfill
    \subfloat{
        \includegraphics[width=0.3\textwidth]{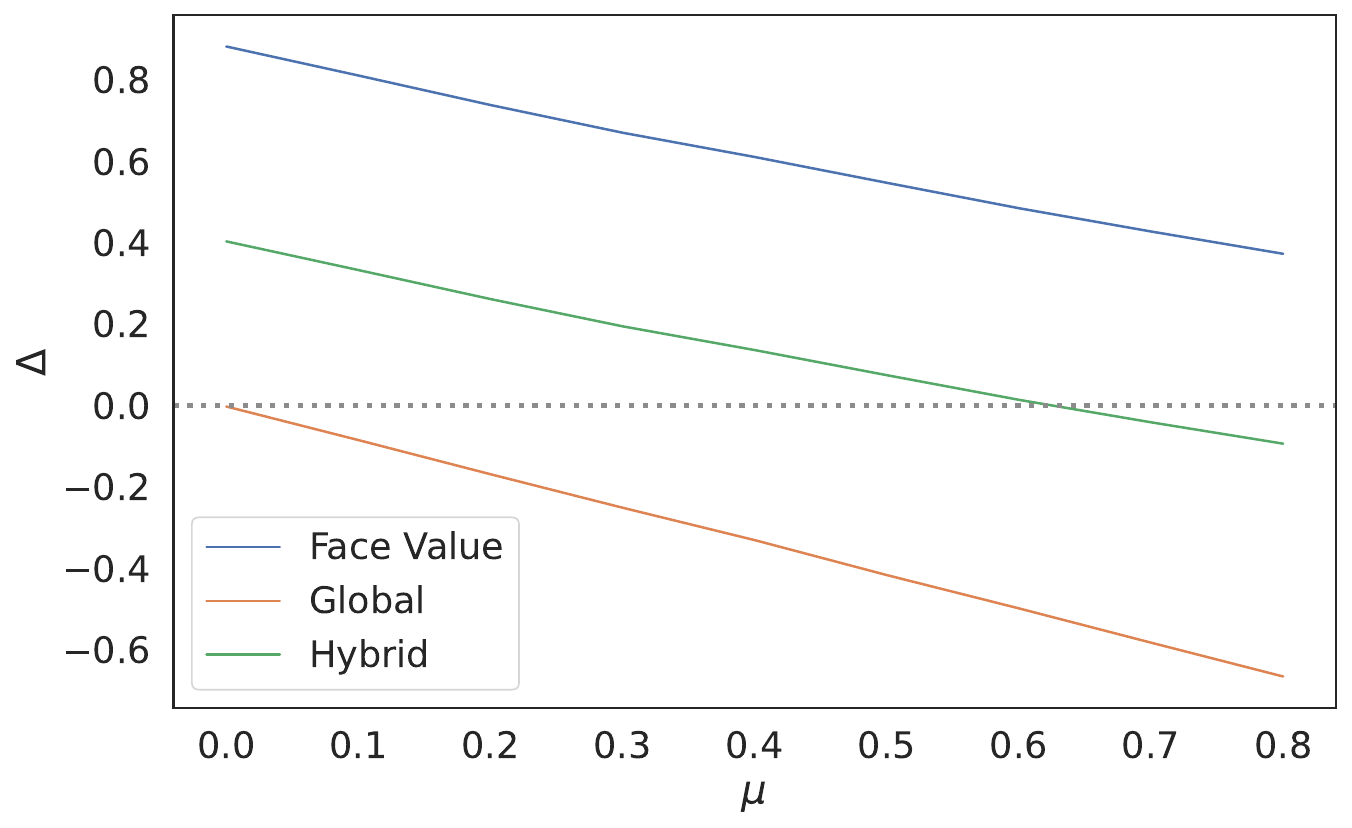}
        \label{subfig:b}
    }\hfill
    \subfloat{
        \includegraphics[width=0.3\textwidth]{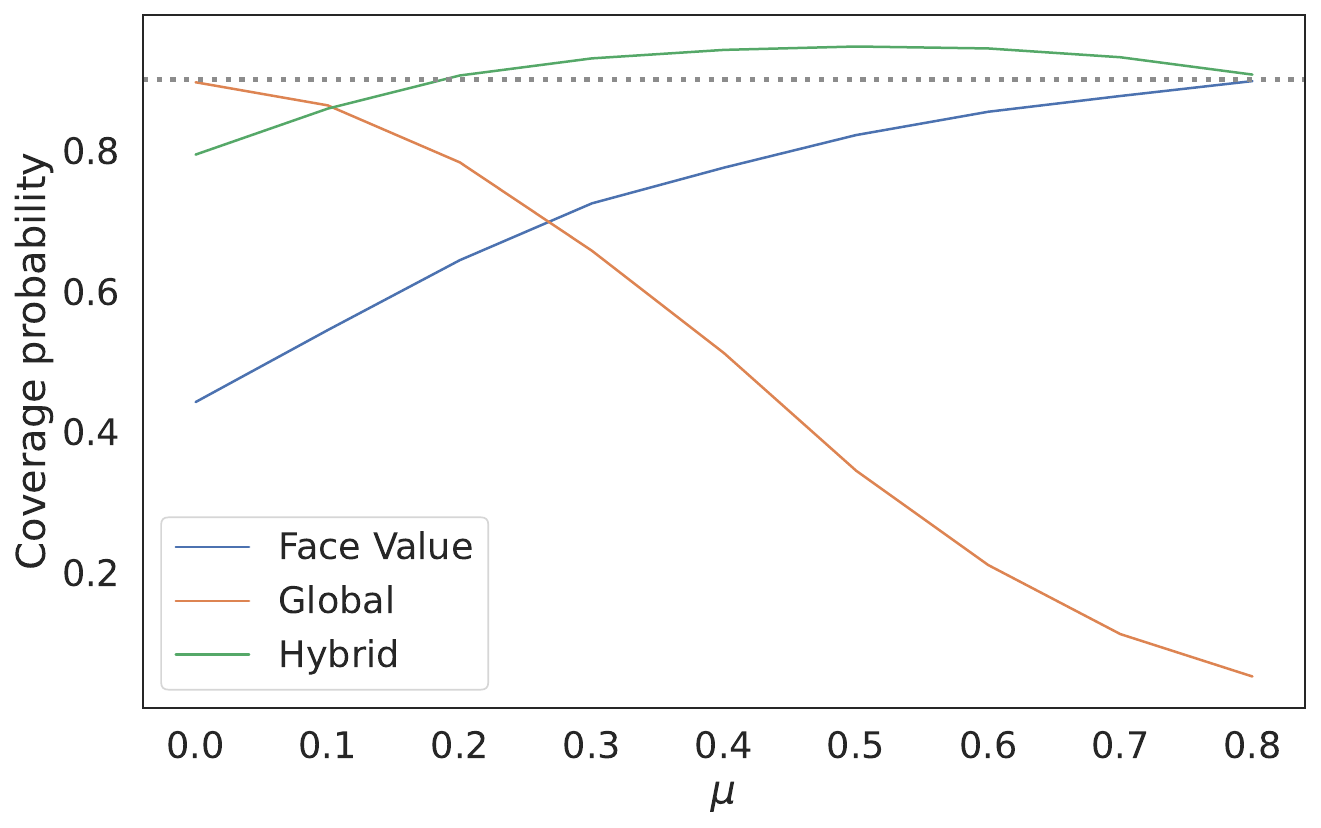}
        \label{subfig:c}
    }\\
    \subfloat{
        \includegraphics[width=0.3\textwidth]{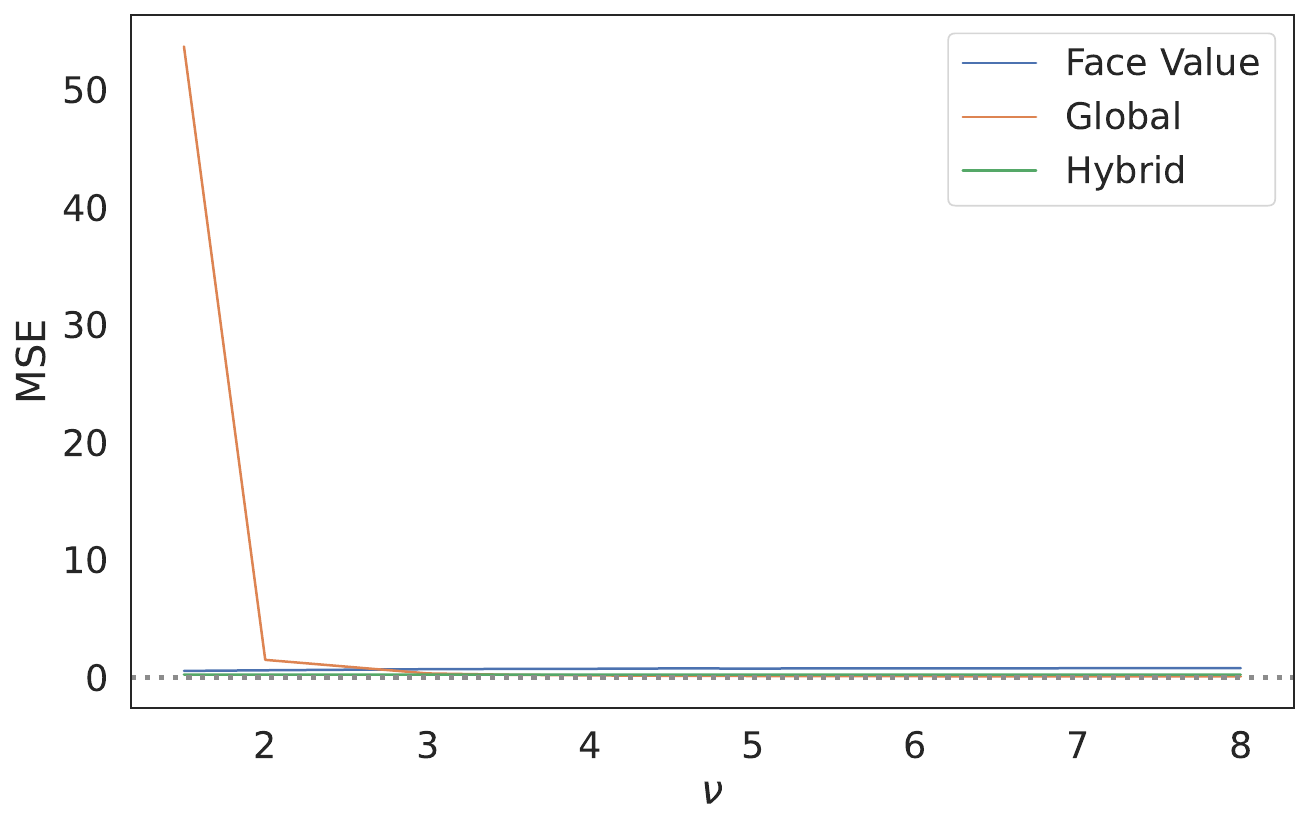}
        \label{subfig:d}
    }\hfill
    \subfloat{
        \includegraphics[width=0.3\textwidth]{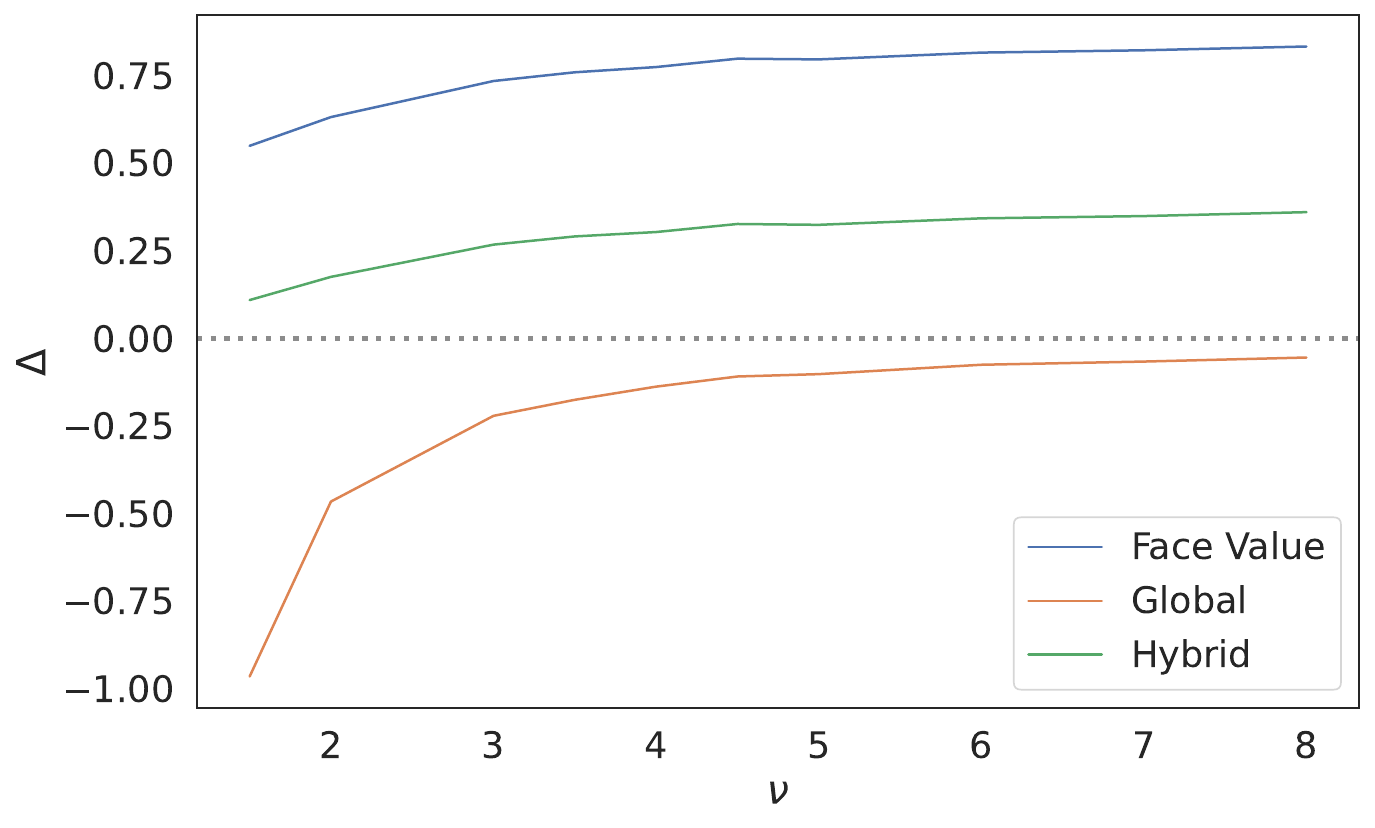}
        \label{subfig:e}
    }\hfill
    \subfloat{
        \includegraphics[width=0.3\textwidth]{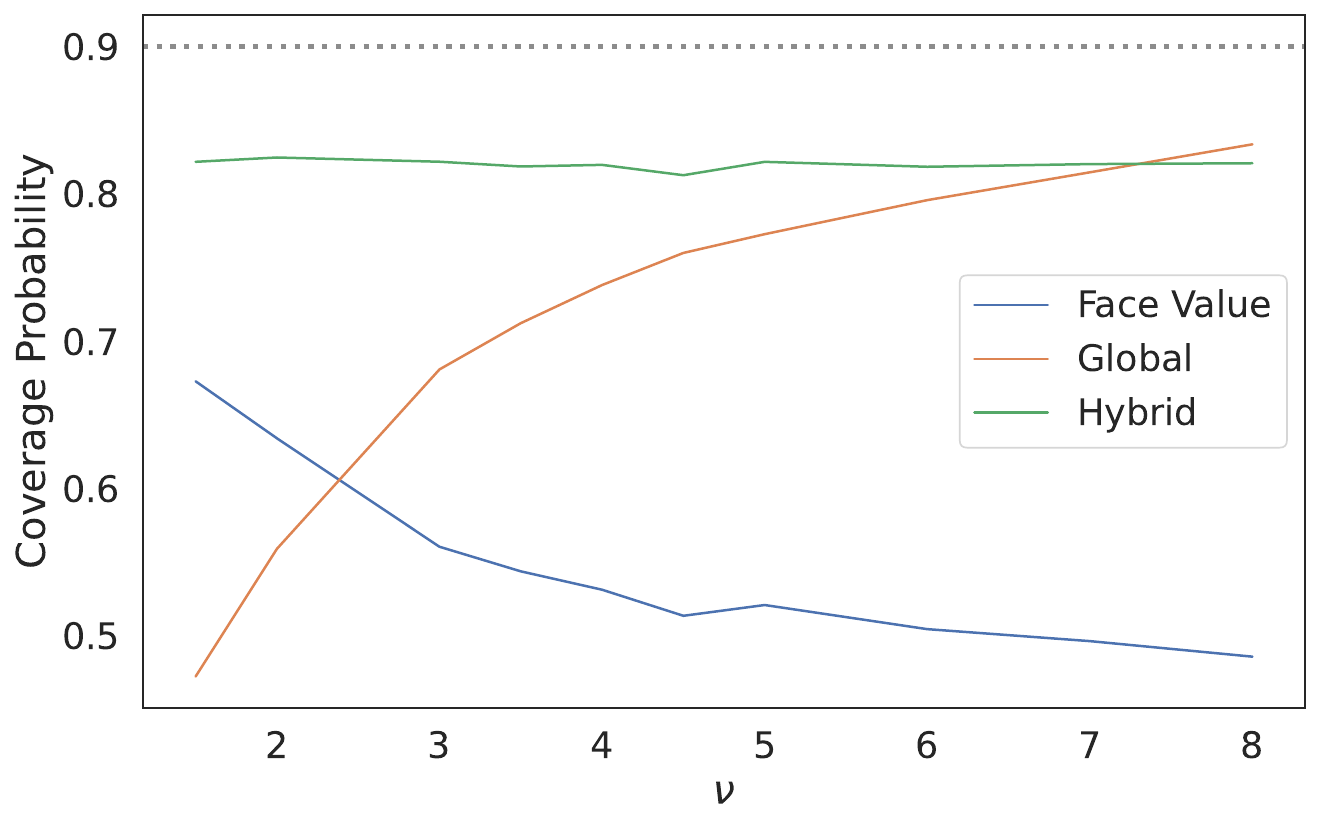}
        \label{subfig:f}
    }\\
    \subfloat{
        \includegraphics[width=0.3\textwidth]{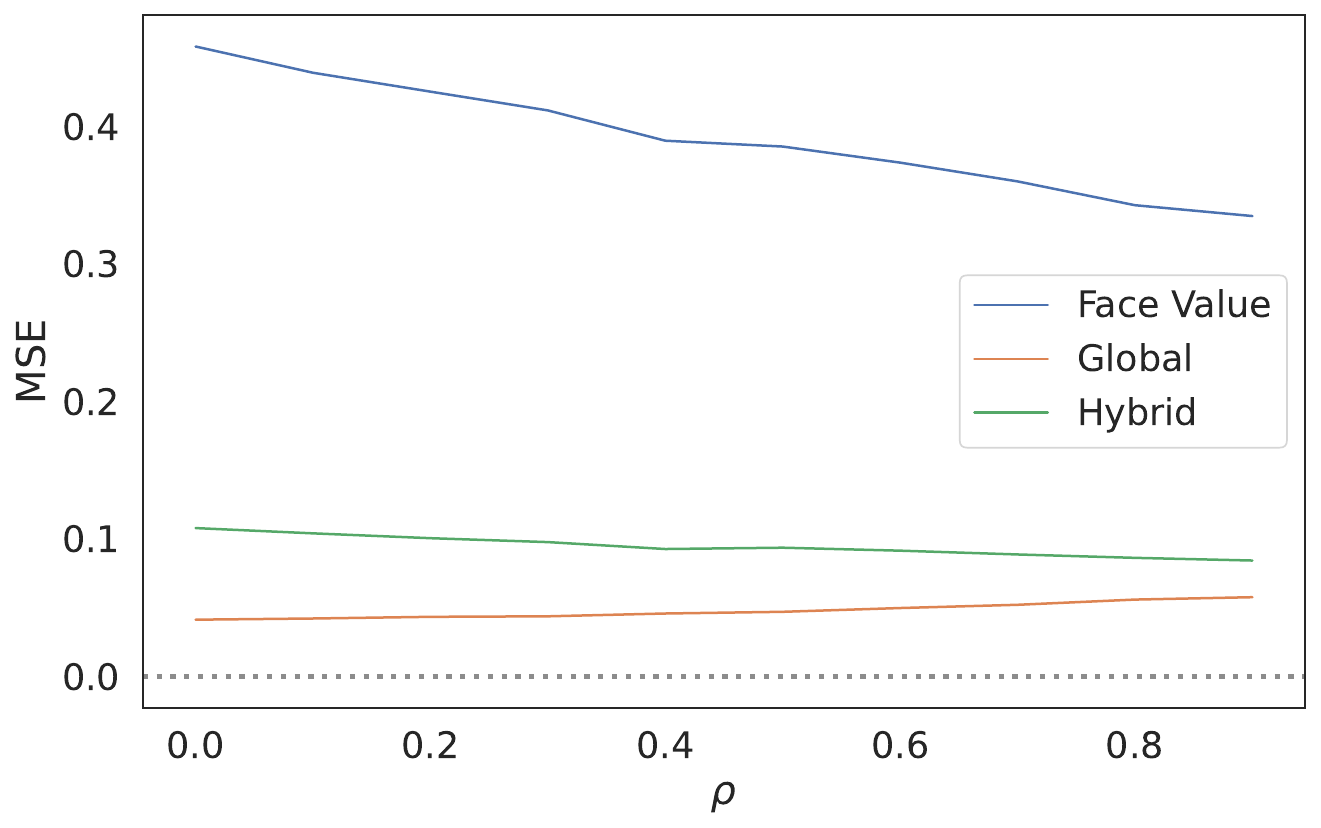}
        \label{subfig:g}
    }\hfill
    \subfloat{
        \includegraphics[width=0.3\textwidth]{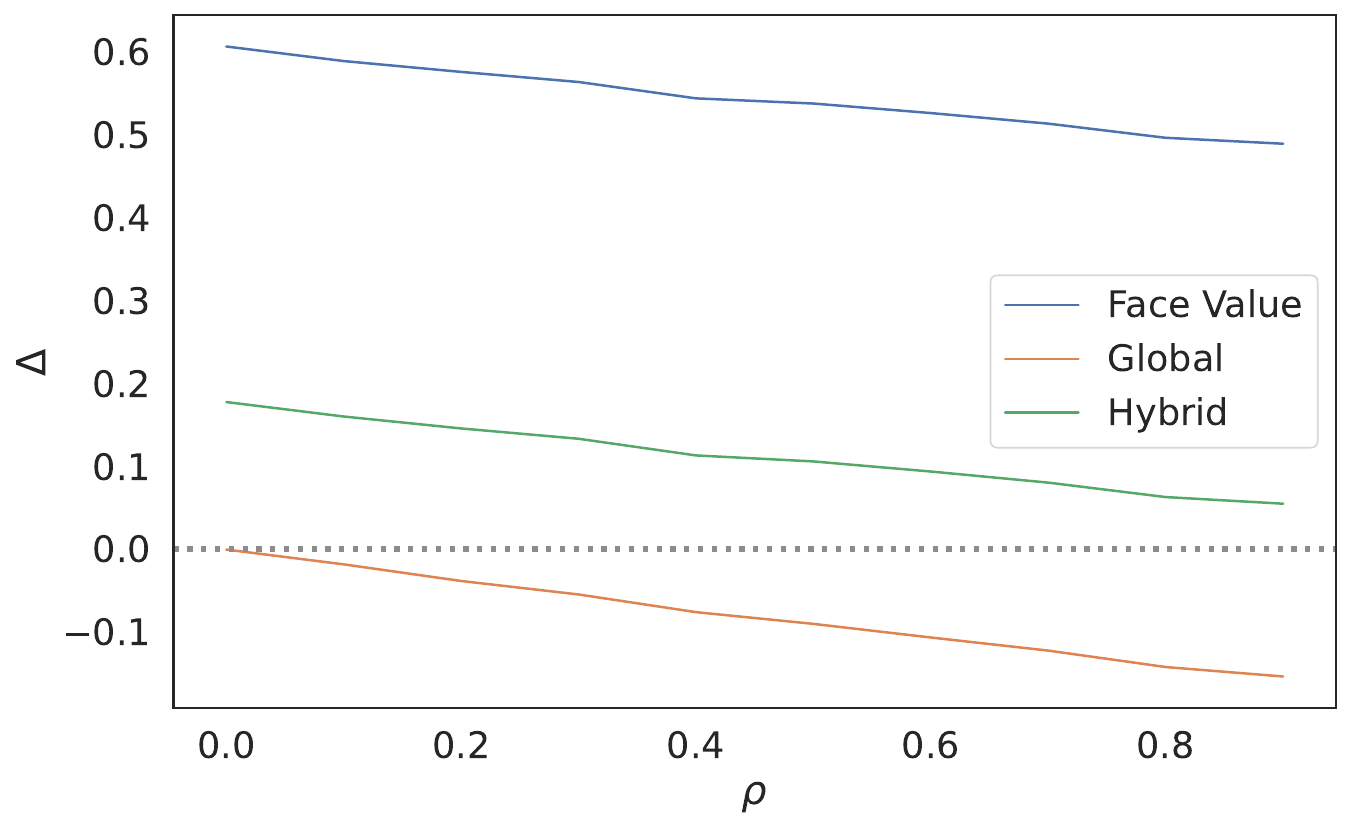}
        \label{subfig:h}
    }\hfill
    \subfloat{
        \includegraphics[width=0.3\textwidth]{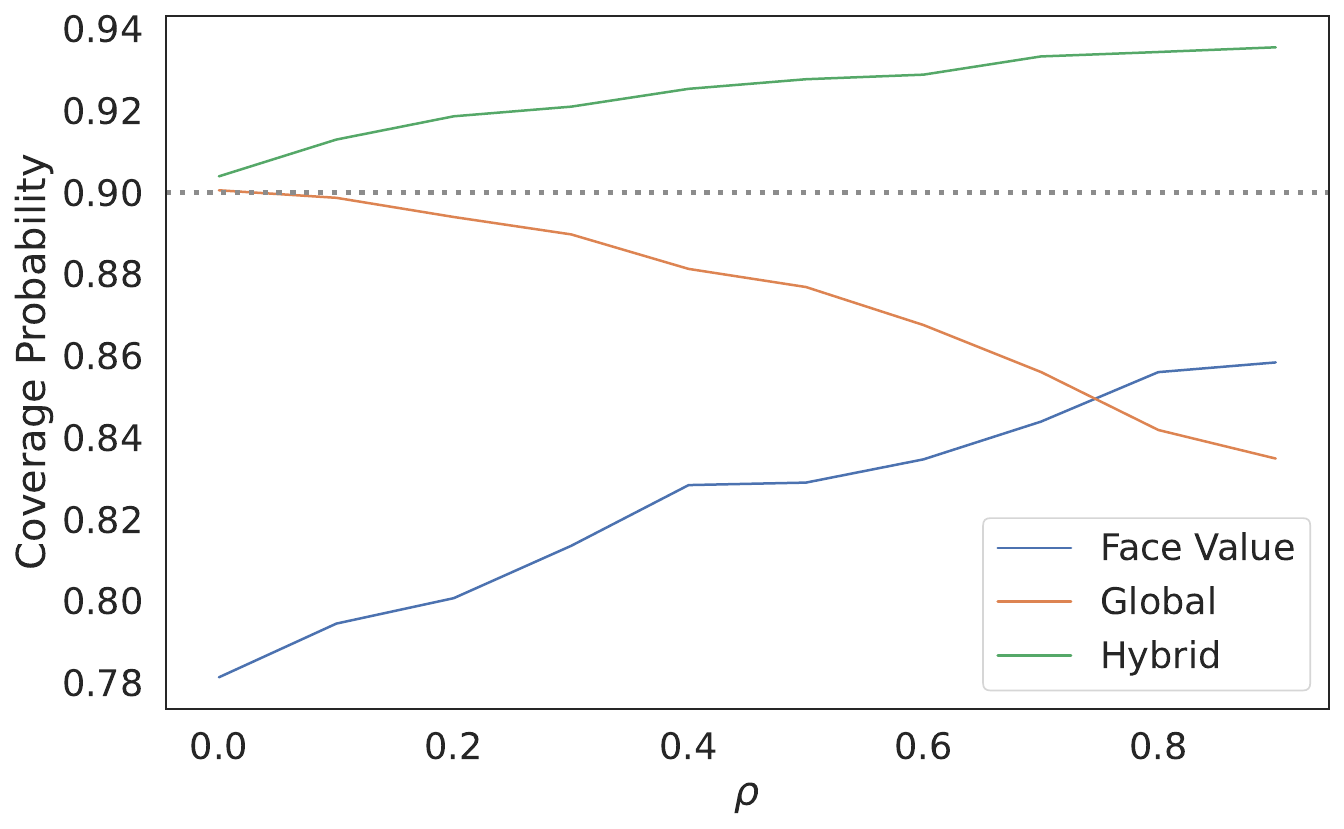}
        \label{subfig:i}
    }

    \caption{MSE (left), bias ($\Delta$, center), and coverage probability (right) for the \textit{Face Value} (blue), \textit{Bayesian Global Shrinkage} (orange), and \textit{Bayesian Hybrid Shrinkage} (green) approaches, as a function of prior mean ($\mu$, top row), degrees of freedom ($\nu$, middle row), and correlation ($\rho$, bottom row).}
    \label{fig:misspecified}
\end{figure*}

\vspace{0.2cm}

\subsection{Empirical Analysis} 

We further evaluate the three approaches in a real-world setting, analyzing 167 experiments with paired replication studies. We compare the mean absolute error (MAE) and the coverage probability for 90\% confidence/credible intervals. The results, summarized in Table~\ref{tbl:summary}, demonstrate that the \textit{Bayesian Hybrid Shrinkage} approach outperforms both the \textit{Face Value} and \textit{Bayesian Global Shrinkage} methods in terms of MAE, while maintaining strong coverage properties.

\begin{table}[h!]
    \centering
    \vspace{-0.1cm}
    \begin{tabular}{||c c c||} 
     \hline
     Model & MAE & Coverage \\ [0.5ex] 
     \hline\hline
     Face Value & 1.280 & 0.92 \\ 
     \hline
     Global Shrinkage & 1.121 & 0.91 \\
     \hline
     Hybrid Shrinkage & 1.012 & 0.91 \\ [1ex] 
     \hline
\end{tabular}
\caption{Summary of Performance Metrics -- Coverage and (MAE x 1000) -- across a collection of real-world experiments with paired replication studies.}
\label{tbl:summary}
\end{table}

\vspace{0.3cm}
\section{Conclusion}
This paper introduce a two stage Bayesian Shrinkage estimator to tackle the Winner's Curse with a scalable strategy for posterior inference. Future research plans include: expanding to additional validation strategies like prior elicitation to achieve specific operating characteristics; guidance for practical implementation of these techniques; offline-evaluation priors~\cite{Radwan2024Eval}; deeper formalization of theoretical properties for posterior inference. We anticipate that deeper exploration along these avenues will enable a wider adoption of these methods in the online experimentation domain.
\bibliographystyle{unsrt}
\bibliography{references}

\end{document}